\documentclass[letterpaper,10pt,conference]{ieeeconf}

\IEEEoverridecommandlockouts

\usepackage{graphicx}
\usepackage{epstopdf}
\usepackage{float}
\usepackage{stfloats}
\fnbelowfloat

\usepackage{amsmath,amssymb}
\usepackage{mathrsfs}
\usepackage{bm}

\usepackage{amsthm}
\usepackage{etoolbox}

\usepackage{acronym}
\usepackage{paralist}
\usepackage{multicol}
\usepackage{tikz}
\usepackage[caption=false,font=footnotesize]{subfig}

\usepackage[dvipsnames]{xcolor}

\usepackage{macros}

\makeatletter
\newcommand{\ClearExistingTheoremEnvironment}[1]{%
  \expandafter\let\csname #1\endcsname\@undefined
  \expandafter\let\csname end#1\endcsname\@undefined
  \expandafter\let\csname c@#1\endcsname\@undefined
  \expandafter\let\csname the#1\endcsname\@undefined
  \expandafter\let\csname p@#1\endcsname\@undefined
}
\ClearExistingTheoremEnvironment{problem}
\ClearExistingTheoremEnvironment{remark}
\ClearExistingTheoremEnvironment{Remark}
\ClearExistingTheoremEnvironment{theorem}
\ClearExistingTheoremEnvironment{Theorem}
\ClearExistingTheoremEnvironment{lemma}
\ClearExistingTheoremEnvironment{corollary}
\ClearExistingTheoremEnvironment{proposition}
\ClearExistingTheoremEnvironment{definition}
\ClearExistingTheoremEnvironment{assumption}
\let\ClearExistingTheoremEnvironment\@undefined
\makeatother

\newtheorem{problem}{Problem}
\AtEndEnvironment{problem}{\hfill$\diamond$}

\AtEndEnvironment{remark}{\hfill$\blacklozenge$}

\AtEndEnvironment{Remark}{\hfill$\blacklozenge$}

\newtheorem{theorem}{Theorem}

\newtheorem{lemma}{Lemma}

\AtEndEnvironment{corollary}{\hfill$\vartriangle$}

\newtheorem{proposition}{Proposition}

\newtheorem{definition}{Definition}
\AtEndEnvironment{definition}{\hfill$\bullet$}

\newtheorem{assumption}{Assumption}
\AtEndEnvironment{assumption}{\hfill$\circ$}

\newcommand{\bequ}{\begin{eqnarray}}
\newcommand{\eequ}{\end{eqnarray}}

\def\BibTeX{{\rm B\kern-.05em{\sc i\kern-.025em b}\kern-.08em
    T\kern-.1667em\lower.7ex\hbox{E}\kern-.125emX}}

\usepackage{xurl}
\usepackage{cite}
\makeatletter
\let\NAT@parse\undefined
\makeatother
\usepackage[colorlinks=true,allcolors=blue]{hyperref}

\begin{document}

\title{\color{Black}
Path Invariance of a Quadrotor System under Cyber Attacks with Theoretical Guarantees
\thanks{This research has been supported by NJIT's startup funds.}
\thanks{Hamza Mahmood$^*$ and Adeel Akhtar$^*$ are with the
Department of Mechanical \& Industrial Engineering at the New Jersey Institute of Technology, Newark, NJ 07102 (email: \texttt{\{hm576, adeel.akhtar\}@njit.edu}). Usman Ali$^\dagger$ is with the Faculty of Computing, Engineering and Media at the De Montfort University, Leicester, LE1 9BH, UK (email: \texttt{usman.ali@dmu.ac.uk})
}}

\author{Hamza Mahmood$^*$ \and Usman Ali$^{\dag}$ \and Adeel Akhtar$^*$
}

\maketitle

\begin{abstract}
This paper presents a path-following controller for a quadrotor system to guarantee safe maneuvers, in terms of forward path invariance, in the presence of cyber-physical attacks. We assume that an adversarial agent can control any one of the rotors through a false data injection (FDI) type of attack. A feedback controller is designed using transverse feedback linearization which guarantees that the system follows a class of smooth curves under FDI attacks. Our proposed controller is computationally efficient, with a closed-form analytical expression, that not only mitigates the effect of bounded malicious signal but also ensures mission success. We provide theoretical guarantees of forward path invariance under FDI attacks with realistic assumptions and demonstrate the effectiveness of our approach through simulation.
\end{abstract}

\section{Introduction}
\label{sec:Introduction}

Unmanned aerial vehicles (UAVs) such as quadrotors have important applications in many areas such as search and rescue, emergency management, package delivery, surveillance, agriculture, and drone photography. In all these applications, a typical mission involves a high-level path planning task of finding an obstacle-free path and then a low-level control task of designing a control policy to track the planned trajectory in a safe manner~\cite{Lav06}. One way of ensuring safety is to render the path forward-invariant~\cite{NieFulMag10,AkhNieWas2015}. Achieving invariance (safe maneuver), however, in the presence of cyber-physical attacks is a challenging problem. The design of a controller that can guarantee forward path invariance in the presence of adversarial attacks is the main focus of this paper.

The problem of trajectory tracking for quadrotors has been considered in~\cite{LiLiuWen2024} in the presence of cyber-physical attacks. However, as highlighted in~\cite{AguHesKok05}, typical trajectory tracking solutions have performance limitations; they cannot guarantee that a system stays on a trajectory once it reaches it. These limitations are overcome in~\cite{AkhWasNie2013,AkhWasNie2012}, where the authors provided an alternative path-following controller that ensures that the system sticks to the path, once it hits the trajectory, by providing guarantees on forward path invariance, and that has essentially better safety properties. These works, however, do not consider adversarial attacks. In this paper, we aim to design a controller that can ensure safety during path-following using forward path-invariance in the presence of cyber-physical attacks.

A problem somewhat similar, but fundamentally different, is considered in~\cite{GargSanAlv2022} where the authors provide a set invariance argument but it is \textit{not} path invariance. Specifically, their controller guarantees that the system stays within a bound of the desired point or trajectory, while in our work, we want to ensure that once the system reaches the curve, it never leaves the curve. Hence, the goal of our controller design is to make the given path asymptotically stable and forward invariant.

In this paper, we consider a particular type of cyber-physical attack in UAVs, namely false data-injection (FDI) attack~\cite{XuHonWe2024,GaoHaoYan2024}, where we consider a quadrotor system with only one vulnerable input and assume that this actuator is hijacked in the sense that the adversarial agent overrides the control input and injects a malicious signal into the actuator. Our goal is then to design a control signal for the remaining actuators to ensure a safe maneuver, i.e., follow a given curve with little or no deviation. Our method does not assume an \textit{a priori} attack model; rather, the attack is detected (for various attack detection mechanisms in CPSs, see~\cite{GonYanWu2023} and references therein) and the instantaneous resulting motor speed is measured. The main theme of this paper is thus to design a controller that not only mitigates the effect of the malicious signal, but also fulfills some mission objectives, especially in terms of safety.

We treat the path-following problem under false data injection (PFP-FDI) as a set stabilization problem~\cite{NieMag04}, where the given path is represented as a set, and then we stabilize the maximal control invariant subset of the path~\cite{NieFulMag10} and render it as an invariant set in the presence of adversarial attack of FDI type. Our solution is based on \textit{transverse feedback linearization (TFL)}~\cite{DsouNie2023}. The problem is challenging because the quadrotor system under adversarial attack is not differentially flat and subsequently has uncontrolled internal dynamics. Despite that, we show, by designing computationally efficient closed-form controllers, that the quadrotor system with one rotor fully controlled by an adversarial agent follows the desired path in a strict sense.

While the details of our controller design and theoretical results can be found in the paper, we summarize our main contributions below:
{

\begin{enumerate}
    \item A well-defined vector relative degree of a quadrotor system in case of one rotor fully controlled by an adversarial agent (Lemma~\ref{lem:welldefined_deg});
    \item A diffeomorphism that geometrically transforms the quadrotor system into a partial linear system (Lemma~\ref{lemma:Diffeo});
    \item Asymptotic stability and invariance of a given path in the presence of adversarial attacks on a single motor (Theorem~\ref{theo:main_result}).
\end{enumerate}
}

{
\subsection{Notation}
\label{subsec:Notation}
We represent a vector $x \in \Real^{n}$ with components $x_1, x_2, \ldots, x_n$ by $x = \col(x_1,\ldots,x_n)$. For Euclidean space, we denote the inner product and norm by $\langle \cdot,\cdot \rangle$ and $\| \cdot \|$, respectively. We abbreviate the trigonometric functions as $c_{\theta} \eqdef \cos\theta,\hspace{1mm} s_{\theta} \eqdef \sin\theta,\hspace{1mm} t_{\theta} \eqdef \tan\theta$. The point-to-set distance of a  point $p \in \Real^n$ from a nonempty set $A\subseteq  \Real^n$ is given by $\norm{p}_A \eqdef \inf_{q \in A}\|p - q\|$. If $h : A \rightarrow B$ and $s: B\rightarrow C$ are two maps, then their composition is denoted by $s\circ h : A \rightarrow C$. The derivative of a $C^{1}$ map $f:\Real^n \rightarrow \Real^m$ at a point $p\in \Real^n$ is written as  $df_p\eqdef \left. \frac{\partial f}{\partial x}\right|_{x = p}$.
If $f,g: \Real^n \rightarrow \Real^n$ and $\lambda: \Real^n \rightarrow \Real$ are smooth maps, the Lie derivatives are denoted by $L_g^0\lambda \eqdef \lambda$, $L_g^k\lambda \eqdef L_g(L_g^{k-1}\lambda)$, $L_gL_f\lambda\eqdef L_g(L_f\lambda)$.
}

\section{Preliminaries}
\label{sec:Some_Preliminaries}
{

In this section, we define a class of parametric curves that can be represented as a zero-level set of a smooth function and review some elementary concepts related to curves; for details, see~\cite{pressley2010}.
\begin{definition}[Regular parameterized curve]
    A parametric curve $\sigma: I\subseteq \Real \to \Real^3$, for some open interval $I$, is called a regular curve if for each $\lambda \in I$, $\sigma^{\prime}(\lambda) \neq 0$, where $\sigma^{\prime}(\lambda)\eqdef {d\sigma}/{d\lambda}$.
\end{definition}
The following proposition shows that any regular curve has a unit-speed parameterization~\cite[Proposition 1.3.6]{pressley2010}.
\begin{proposition}[Unit-speed parameterization of a regular curve]
A parameterized curve has a unit-speed reparameterization if and only if it is
regular.
\end{proposition}
 This means that every regular curve $\sigma$ has a parameterization corresponding to which we have $\|\sigma^\prime(\cdot)\| \equiv 1$. In the sections to follow, we shall always assume that a regular curve is unit-speed parameterized.
\begin{definition}[$L$-periodic curves]
\label{def:L_periodic_curves}
    Let $\sigma: \Real \to \Real^3$ be a parametric curve and let $L\in\Real$. The curve $\sigma$ is $L$-periodic if $\sigma(\lambda+ L) = \sigma(\lambda)$ for all $\lambda\in\Real$. If $\sigma$ is not constant and is $L$-periodic for some $L\neq 0$, then $\sigma$ is said to be closed.
\end{definition}

}

\section{Mathematical Model of a Quadrotor}
\label{sec:Mathematical_Model_Quadrotor}
{
In this paper, we consider a standard quadrotor model given in~\cite{AkhWasNie2012}. Let us consider two right-handed orthonormal frames in $\Real^3$, denoted by $\mc{I} = \set{\Vec{i}_1,\Vec{i}_2,\Vec{i}_3}$ and $\mc{B} = \set{\Vec{b}_1,\Vec{b}_2,\Vec{b}_3}$, called the inertial and body-fixed reference frames, respectively. The inertial frame $\mc{I}$ is fixed, i.e., it neither translates nor rotates. The body-fixed frame, as the name suggests, is attached to the center of gravity of the quadrotor body. The orientation of the quadrotor, also known as \textit{attitude}, is the orientation of the frame $\mc{B}$ with respect to the inertial frame $\mc{I}$. The set of all possible orientations is given by
$
\SO3 \eqdef \set{R\in\Real^{3\times 3}: R^\top R = I = R R^{\top}, \det(R) = +1},
$
which constitutes a Lie group, i.e., $\SO3$ is a smooth manifold and forms a group under matrix multiplication. The elements of $\SO3$ are called rotation matrices. One possible way to represent rotation matrices is by the three \textit{ZYX Euler angles}, namely yaw ($\psi \in \Real$), pitch ($\theta \in \Real$), and roll ($\phi \in \Real$). To obtain $\mc{B}$, we rotate $\mc{I}$

using the rotation matrix
\begin{equation}
\label{eq:roation_matrix_product}
    R^{i}_{b}(\Phi) = \begin{bmatrix}
    c_\psi & -s_\psi & 0 \\
    s_\psi & c_\psi & 0 \\
    0 & 0 & 1
\end{bmatrix}
\begin{bmatrix}
    c_\theta & 0 & s_\theta \\
        0 & 1 & 0 \\
    -s_\theta &0 &  c_\theta
\end{bmatrix}
\begin{bmatrix}
    1 & 0 & 0\\
    0 & c_\phi & -s_\phi\\
    0 & s_\phi & c_\phi
\end{bmatrix},
\end{equation}
where the three Euler angles are denoted by $\Phi \eqdef (\phi,\theta,\psi)$.
We define the domain of Euler angles with $\mc{E}_D \eqdef \set{ \Phi = (\phi,\theta,\psi) \in (-\pi,\pi)\times(-\pi,\pi)\times(-\pi,\pi): \cos \theta > 0}$. It is well-known that the map $R^{i}_{b}: \mc{E}_D \to R^{i}_{b}(\mc{E}_D)$ is a diffeomorphism. The inverse map $(R^{i}_{b})^{-1}$ is a coordinate chart of the Lie group $\SO3$. While this representation is unique, it is not global, with singularities at angles $\theta = \pm \pi/2$. We define $Y(\Phi)$ and $Y^{-1}(\Phi)$ as
\begin{equation}
    \label{eq:DCM}
    Y(\Phi) = \begin{bmatrix}
    1 & 0 & -s_\theta\\
    0 & c_\phi & c_\theta s_\phi\\
    0 & -s_\phi & c_\theta c_\phi
\end{bmatrix},\hspace{0.3mm}
Y^{-1}(\Phi) = \begin{bmatrix}
    1 & s_\phi t_\theta & c_\phi t_\theta\\
    0 & c_\phi & -s_\phi\\
    0 & s_\phi/c_{\theta} & c_\phi/c_{\theta}
\end{bmatrix}.
\end{equation}
Using the above relationship, the kinematic equation of a rotating rigid body can be described as $\dot \Phi = Y^{-1}(\Phi)\Omega$, where $\Omega(t) \eqdef (p(t),q(t),r(t)) \in \Real^3$ are body rates of the rotating rigid body. Let $x_q(t) \eqdef (x_I(t), y_I(t), z_I(t)) \in \Real^3$ denote the quadrotor's position, and $v_q(t) = (v_x(t), v_y(t), v_z(t))\in\Real^3$ be its velocity with respect to the inertial frame $\mc{I}$. Let $J \eqdef \diag(I_x,I_y,I_z) \in \Real^{3\times3}$ be the inertia matrix of the quadrotor expressed in body-fixed frame $\mc{B}$.

Let $m$ and $g$ be the mass of the quadrotor and acceleration due to gravity, respectively.
The input torque about the three body axes is denoted by $\tau(t) \eqdef \col(\tau_p(t), \tau_q(t), \tau_r(t))\in\Real^3$, which, together with the combined thrust $u_f(t)$, defines the four control inputs of the UAV system. The translational and rotational
drag coefficients are denoted by $k_t$ and $k_r$, respectively, and are assumed to be the same in all
directions.

The dynamics of the quadrotor can be locally represented\footnote{note that the argument $t$ is dropped for notational simplicity} using Euler angles as
\begin{equation}
\label{eq:quad-eular}
\begin{aligned}
  \dot \Phi &= Y^{-1}(\Phi) \Omega,\\
  \dot \Omega &= J^{-1} \left( \tau - k_r \Omega - (\Omega \times J\Omega)\right),\\
   \dot x_{q} &= v_{q},\\
  \dot{v}_q &= \frac{u_f}{m} R^{i}_{b}(\Phi) \Vec{b}_3 - g \Vec{i}_3 - \frac{k_t}{m} v_q.
\end{aligned}
\end{equation}

It can be seen that the input of the quadrotor model~\eqref{eq:quad-eular} is $\mc{U}_\tau(t) \eqdef (u_f(t),\tau(t))\in\Real^4$. However, the inputs of a physical quadrotor are forces generated by each propeller. Let the force generated by the $i$-th propeller be $f_i(t)\in\Real$, and we denote all forces of the quadrotor by $\mc{U}_f(t) \eqdef (f_1(t),f_2(t),f_3(t),f_4(t))\in\Real^{4}$. We assume that these thrust forces are nonnegative, twice differentiable, and bounded above by some $M < \infty$. The force inputs $\mc{U}_f$ are related to $\mc{U}_\tau$ by an invertible linear map, whose matrix representation is given by
\begin{equation}
\label{eq:Uf_and_Utau_relation}
    \begin{bmatrix}
        u_f\\ \tau_p\\ \tau_q \\\tau_r
    \end{bmatrix} = \begin{bmatrix}
        1 & 1 & 1 & 1\\
        0 & -l & 0 & l\\
        -l & 0 & l & 0\\
        d & -d & d & -d
    \end{bmatrix} \begin{bmatrix}
        f_1 \\ f_2\\ f_3 \\ f_4
    \end{bmatrix},
\end{equation}where $l$ is the distance of each rotor from the quadrotor's center of gravity and
$d$ is the ratio between the drag and the thrust coefficients of
the blade.

{
Let $x(t) = \col(\Phi(t), \Omega(t), x_q(t), v_q(t)) \in \Real^{12}$ be the state of the quadrotor. We take the output $y\in\Real^3$ of the system~\eqref{eq:quad-eular} to be the position $x_q$ of the quadrotor
\begin{equation}
\label{eq:output}
    y = h(x) = x_q.
\end{equation}

Next, we consider a quadrotor model under an adversarial attack.
}
{
\subsection{Quadrotor under Adversarial Attack}
\label{subsec:Under_Adv_Attack}
 We assume that only one rotor is vulnerable and without loss of generality, let us assume that $f_2$ is the vulnerable input. Meanwhile, the other three control inputs $f_1,f_3,f_4$ are secure and cannot be attacked. We also assume that the adversarial agent completely controls the second motor and can inject a malicious FDI signal, i.e., the thrust force generated by the second motor is now fully under the control of the adversarial~agent.

\begin{definition}[Adversarial Signal]
\label{def:adversarial_signal}
    An adversarial signal is a smooth bounded map $\mc{A}:\Real_{\geq 0}\times \Real^3 \to [0,M_a] \subset \Real$, $(t,x_q)\mapsto \mc{A}(t,x_q)$ for some $M_a<\infty$.
\end{definition}
We assume that there exists an attack detection mechanism that detects the instantaneous value of the adversarial signal~$\mc{A}(t,x_q)$; its value for each instant $t \ge 0$ is known to us.
Since actuator two is under attack, it follows from Definition~\ref{def:adversarial_signal}
that $f_2 = \mc{A}(t,x_q)$. By substituting $f_2 = \mc{A}(t,x_q)$ in~\eqref{eq:Uf_and_Utau_relation}, we obtain
\begin{equation}
\label{eq:tau_p}
    \tau_p =  \left(u_f - \frac{\tau_r}{d}\right)\left(\frac{l}{2}\right) - 2l\mc{A}(t,x_q).
\end{equation}
Under this adversarial attack on the second motor, the thrust forces $f_1,f_3,f_4$, can alternatively be represented by $(u_f,\tau_q,\tau_r)$  using~\eqref{eq:Uf_and_Utau_relation}. Before formally stating the problem, we make the following mild assumption:
\begin{assumption}
\label{as:uf_nonzero_&_no gimbal lock situation}
    For all time $t \geq 0$, the combined thrust force $u_f \neq 0$ and $\phi, \theta \neq \pm \pi/2$, i.e., the system does not go into gimbal lock situation.
\end{assumption}
Note that this assumption is reasonable: in order for the quadrotor to be in the air, it needs to have a nonzero combined thrust force. Also, we highlight that the singularity associated with gimbal lock is an artifact of the local parameterization, i.e., Euler angles, on $\mathsf{SO}(3)$.

}

\section{Problem Formulation}
\label{sec:Problem_formulation}
{

We consider a regular, smooth curve $\gamma$ in $\Real^{3}$ as the desired path for our quadrotor to follow. Let $\sigma:\Real\to\Real^{3}$ be the unit-speed parameterization of $\gamma$. The curve $\gamma$ can be either closed (see Definition \ref{def:L_periodic_curves}), e.g., a circle, or not closed, e.g., a straight line. It can be shown that  there exists a smooth map $s:\Real^{3}\to\Real^{2}$, $s \eqdef (s_1, s_2)$, such that $\gamma = s^{-1}(0)$ with $\textup{rank}(ds_{y}) = 2$ for each point $y$ on $\gamma$. Since $\gamma = s^{-1}(0)$, the path in the output space~\eqref{eq:output} can be represented by its zero-level set representation as
\begin{equation}
\label{eq:little_gamma}
\gamma \eqdef s^{-1}(0) = \set{ y \in \Real^3 : s_1(y) = s_2(y) = 0 }\subset \Real^3.
\end{equation}
We lift the path $\gamma$ to the quadrotor's state space:
\[\Gamma \eqdef \Bigl\{x \in \Real^{12} : s_1(h(x)) = s_2(h(x)) = 0\Bigl\}\hspace{0.7mm}.\]
We can now state the path-following control problem in the presence of adversarial agent:

\begin{problem}\label{problem:globalinvariance}
Consider a quadrotor model~\eqref{eq:quad-eular} satisfying Assumption~\ref{as:uf_nonzero_&_no gimbal lock situation}, an adversarial agent $\mc{A}$ given in Definition~\ref{def:adversarial_signal}, a desired path $\gamma$ defined in~\eqref{eq:little_gamma}, and a set of initial conditions in the neighborhood of ~$\Gamma$. Design a controller $\kappa:\Real^{12}\times \Real^k\to\Real^3$ (where $k\in\mbb{Z}^{+}$ is the dimension of the augmented state) such that the output of the system converges to the set $\gamma$, i.e., $\norm{y(t)}_{\gamma} \to 0$, as $t\to\infty$, such that $\Gamma$ or some subset of $\Gamma$ becomes asymptotically stable and forward invariant.
\end{problem}

It should be noted that the output $y\in\Real^3$ approaches the curve $\gamma$ if and only if the state $x$ of the system approaches $\Gamma$~\cite{NieMag04}. However, in general, the set $\Gamma$ cannot be made forward invariant, so we aim to stabilize a subset of $\Gamma$ and denote it by $\Gamma^{\star}$ (see~\cite{AkhWasNie2012, AkhWasNie2013}).

We observe that the path invariance of a set is a necessary condition for its stability. This means that asymptotic stability of a set implies its forward invariance. Next, we present our path-following~controller.

\section{Path-Following Controller Design}
\label{sec:path_following_controller_design}
{

In general, we cannot stabilize the set $\Gamma$, as highlighted previously, and we seek to stabilize the largest controlled
invariant subset $\Gamma^{\star}$ contained in $\Gamma$. We call $\Gamma^{\star}$ the path
following manifold (see~\cite{AkhWasNie2012, NieFulMag10, ConMagNieTos2010}). For any trajectory of the system that is in the path-following manifold, there exists a suitable choice of control input such that the output associated with the trajectory can be made to remain on the desired path. By rendering $\Gamma^{\star}$ attractive and invariant, Problem~\ref{problem:globalinvariance} can be solved. To obtain $\Gamma^{\star}$, by using the map $s:\Real^3 \to \Real^2$, we define the following:
\begin{equation}
\begin{aligned}
  \alpha =
  \left[\begin{array}{c}\alpha_1(x)\\\alpha_2(x)\end{array}\right]
  \eqdef s \circ h(x)= \left[\begin{array}{c}s_1 \circ h(x)\\s_2 \circ
      h(x)\end{array}\right].
\end{aligned}
\label{eq:alpha_fn}
\end{equation}

We see that $\alpha$ in~\eqref{eq:alpha_fn} has only two component functions. Since we have only three control inputs for our attacked quadrotor, we augment $\alpha$ with one additional function to have the number of component functions of the output equal to the number of control inputs. We then check if the augmented output has a well-defined vector relative degree. For this, let $\pi:\Real^{3}\to\Real$ be any smooth real-valued function. It can be easily shown (see~\cite{AkhWasNie2012}) that the virtual output $\hat{y} = (\alpha, \pi \circ h)$ does not have a well-defined relative degree
because the decoupling matrix is not full rank.

As in~\cite{AkhWasNie2013}, we can solve this problem by delaying the appearance of the control input $u_f$ using two integrators which are included through two additional states: $x_{13} \eqdef u_f, x_{14} \eqdef \dot u_f.$
Define $u_d \eqdef \ddot u_f$ and $u \eqdef \textup{col}(u_d,u_r,u_q) \in \Real^{3}$, where $u_r \eqdef \tau_r$ and $u_q \eqdef \tau_q$. We write the extended model, with a slight abuse of notation, as
\begin{equation}
\label{eq:extended_system}
\dot x = f(x) + g_1(x)u_d + g_2(x)u_r + g_3(x)u_q,
\end{equation}
where $x = \col(\Phi, \Omega, x_q, v_q,x_{13}, x_{14}) \in \Real^{14}$, and $f(x)$ is as follows:

\begin{equation}
\label{eq:f(x)}
f(x) =
\begin{bmatrix}
Y^{-1}(\Phi) \Omega\\
    -\frac{k_r p - lx_{13}/2 +2\mc{A}(t,x_q)l-I_y q r +I_z q r}{I_x}\\
    -\frac{k_r q + I_x p r - I_z p r}{I_y}\\
    -\frac{k_r r - I_x p q + I_y p q}{I_z}\\
    v_{q}\\
\frac{x_{13}}{m} R^{i}_{b}(\Phi) \Vec{b}_3 - g \Vec{i}_3 - \frac{k_t}{m} v_q\\
    x_{14}\\
    0
\end{bmatrix}.
\end{equation}

The vectors $g_1(x),\hspace{0.5mm} g_2(x),$ and $g_3(x)$ are as follows:
\begin{equation}
\label{eq:g's}
\begin{aligned}
 g_1(x) &=
\col(0, 0, 0, 0, 0, 0, 0, 0, 0, 0, 0, 0, 0, 1), \\
g_2(x) &=
\col(0, 0, 0, -l/(2 I_x d), 0, 1/I_z, 0, 0, 0, 0, 0, 0, 0, 0),\\
g_3(x) &=  \col(0, 0, 0, 0, 1/I_y, 0, 0, 0, 0, 0, 0, 0, 0, 0).
\end{aligned}
\end{equation}

To obtain the largest control invariant subset of $\Gamma$, we apply the zero dynamics algorithm~\cite{Isi95} and obtain
\begin{equation}
\label{eq:path_following_manifold}
\Gamma^{\star} = \set{ x \in \Real^{14} : L_{f}^{i}\alpha(x) = 0,\hspace{1mm} i = 0,1,2,3,4 }\subset \Gamma.
\end{equation}
}

{
\label{subsec:vec_rel_deg}
Having a well-defined vector relative degree is key to our controller design. To achieve a well-defined vector relative degree, we choose a function similar
to the one used in~\cite{AkhWasNie2012, AkhWasNie2013}. Let
$U_{\gamma} \subset \Real^3$ be a tubular \nbhd of the path
$\gamma$ and define the function

\begin{equation}
\begin{aligned}
\varpi :\hspace{0.5mm} &U_{\gamma} \rightarrow \Real \\ &y \mapsto \arg
\inf_{\lambda \in \Real}\|y - \sigma(\lambda)\|.
\end{aligned}
\label{eq:proj}
\end{equation}

The function defined above is smooth as long as $U_{\gamma}$ is a
sufficiently small ``tube'' around the curve $\gamma$.
With this definition, the
virtual output function $\hat{y}$ becomes
\begin{equation}
\label{eq:virtual_output_fn}
\begin{aligned}
\hat{y} = \left[\begin{array}{c}\alpha_1(x_q)\\\alpha_2(x_q)\\ \pi(x_q) \end{array}\right] =
\left[\begin{array}{c}s_1 \circ h(x) \\s_2 \circ h(x) \\ \varpi \circ h(x)\end{array}\right].
\end{aligned}
\end{equation}

We now prove that at each point of $\Gamma^{\star}$, the extended quadrotor system~\eqref{eq:extended_system} under the bounded adversarial attack $\mc{A}$ has a well-defined vector relative degree.

}

{
\begin{lemma}
\label{lem:welldefined_deg}
    Given the extended system in~\eqref{eq:extended_system} with a virtual output map defined in~\eqref{eq:virtual_output_fn}, and an adversarial attack signal $\mc{A}$ defined in Definition \ref{def:adversarial_signal}, the system has a well-defined vector relative degree of $\{4,4,4\}$ for all $x\in\Gamma^\star \cap \{x \in \Real^{14}:
x_{13} \ne 0 \}$ under Assumption~\ref{as:uf_nonzero_&_no gimbal lock situation}.
\end{lemma}
\begin{proof}
Consider any arbitrary $x^\star \in \Gamma^\star \cap \{x \in \Real^{14}: x_{13} \ne 0
  \}$. Since $\Gamma^\star
  \subseteq \Gamma$, $x^{\star} \in \Gamma$ and so, the output $h(x^{\star})$ lies on the path $\gamma$. This implies that there exists a $\lambda^{\star} \in \Real$ such that $h(x^\star) = \sigma(\lambda^\star)$. By the definition of vector
  relative degree, we first show that
$
L_{g_i}L^j_f\pi(x) = L_{g_i}L^j_f\alpha_k(x) \equiv 0
$
for $i \in \{1, 2, 3\}$, $j \in \{0, 1, 2\}$, $k \in \{1, 2\}$ in a
\nbhd of $x^\star$. It is easy to see that this holds by simply computing these Lie derivatives. Then, we need to show that the $3 \times 3$ decoupling matrix
\begin{equation}
\label{eq:decoupling_matrix_general}
 D(x^\star)=
  \left[\begin{array}{c c c}
      L_{g_{1}}L_{f}^{3}\alpha_1(x^\star) &  L_{g_{2}}L_{f}^{3}\alpha_1(x^\star) & L_{g_{3}}L_{f}^{3}\alpha_1(x^\star)\\
      L_{g_{1}}L_{f}^{3}\alpha_2(x^\star) &  L_{g_{2}}L_{f}^{3}\alpha_2(x^\star)  & L_{g_{3}}L_{f}^{3}\alpha_2(x^\star)\\
      L_{g_{1}}L_{f}^{3}\pi(x^\star) &  L_{g_{2}}L_{f}^{3}\pi(x^\star) & L_{g_{3}}L_{f}^{3}\pi(x^\star)
  \end{array}\right],
\end{equation}is invertible.
Interestingly, the adversarial attack signal $\mc{A}$ does not appear in any entry of the decoupling matrix $D(x^{\star})$ and so, the determinant of $D(x^{\star})$ is as follows:
\begin{equation}
\begin{split}
\det{(D(x))}  = \frac{ l(x_{13})^2 }{I_x I_y  m^3 d}
  \left\langle -d_\chi\alpha_1,  \left( d_\chi\alpha_2 \times
      \sigma^{\prime} \right) \rangle\right.
      ,
\end{split}
\label{eq:detD}
\end{equation} where $d_\chi \alpha_{i} \eqdef \col(\frac{\partial\alpha_i}{\partial x_I},\frac{\partial\alpha_i}{\partial y_I},\frac{\partial\alpha_i}{\partial z_I})$ for $i \in \{1,2\}$ and $\sigma^\prime \eqdef \col(\frac{d\sigma_x}{d\lambda},\frac{d\sigma_y}{d\lambda},\frac{d\sigma_z}{d\lambda})$. The determinant is zero if and only if any term in the numerator
of~\eqref{eq:detD} is zero. Since $x^\star \in \Gamma^\star \cap \{x \in \Real^{14}:
x_{13} \neq 0 \}$, the combined thrust $x_{13} \neq 0$ under Assumption~\ref{as:uf_nonzero_&_no gimbal lock situation}. Also, $\Sp\{d_\chi \alpha_1, d_\chi
\alpha_2, \sigma^\prime \}(x^\star) = \Real^3$ (see\cite[Lemma V.2]{AkhWasNie2013}) and so, $ \langle -d_\chi\alpha_1, \left(
  d_\chi\alpha_2 \times \sigma^{\prime} \right) \rangle \neq 0$ at
$x^\star$ (see\cite[Lemma V.1]{AkhWasNie2013}). Hence, we have shown that for each $x^\star \in \Gamma^\star
\cap \{x \in \Real^{14}: x_{13} \neq 0 \}$, $\det{(D(x^\star))}\neq
0$, and therefore, the extended system has a well-defined vector relative degree at $x^\star$.
\end{proof}
}

{

The well-defined vector relative degree of $\{4,4,4\}$ for the extended system means that the
internal dynamics has dimension $14 - (4+4+4) = 2$. So, we require two additional functions to define a complete coordinate transformation. }

{
\begin{lemma}[Diffeomorphism]
\label{lemma:Diffeo}
   Suppose Assumption~\ref{as:uf_nonzero_&_no gimbal lock situation} holds, then for each $x^\star \in \Gamma^\star$, there exists a \nbhd $U_x \subset
  \Real^{14}$ of $x^\star$ such that the map $\mathscr{T} : U_{x}
 \rightarrow \mathscr{T}(U_x) \subset \Real^{14}$, defined by
\begin{equation}
\left[\begin{array}{c}\xi_{ji}\\\eta_{i}\\ \mu_{k}\end{array}\right]
= \mathscr{T}(x) =
\left[\begin{array}{c}L^{i-1}_f\alpha_j(x)\\L^{i-1}_f\pi(x)\\ \mu(x)\\ \end{array}\right],
\label{eq:diffeo}
\end{equation}
for $i\in\{1,2,3,4\}$, $j\in\{1,2\}$ and $k\in\{1,2\}$, is a
diffeomorphism.
\end{lemma}
\begin{proof}
   Since the dimension of the internal dynamics is 2, we choose two real-valued
   functions $\mu_1,\mu_2$ to complete the definition of $\mathscr{T}$ similar to~\cite{AkhWasNie2013}\hspace{0.5mm}: \vspace{-0.1in}
 \begin{equation}
  \begin{aligned}
  \mu_1 & \coloneqq \psi,\quad
  \mu_2  \coloneqq -\frac{p}{I_z} -\frac{l r}{2 I_x d}.
  \end{aligned}
\label{eq:internal_states}
 \end{equation}
 Next, we check the rank of the Jacobian matrix $\frac{\partial \mathscr{T}}{\partial x}$ to see whether it defines a local diffeomorphism. The determinant
 of $\frac{\partial \mathscr{T}}{\partial x}$ is given by \vspace{-0.1in}
\begin{equation}
\begin{split}
  \frac{ -l(x_{13})^4 c_{\phi}}{2 I_x m^6 d}
  \langle -d_\chi\alpha_1,  \left( d_\chi\alpha_2 \times \sigma^{\prime} \right) \rangle.
\end{split}
\label{eq:detT}
\end{equation}
Now, this determinant equals
zero if and only if any term in the numerator equals zero. However, by our hypothesis and Assumption \ref{as:uf_nonzero_&_no gimbal lock situation}, $x_{13} \neq 0$ and  $\phi \neq \pm\frac{\pi}{2}$. Moreover, as in the proof of Lemma~\ref{lem:welldefined_deg}, we have $ \langle -d_\chi\alpha_1, \left(
  d_\chi\alpha_2 \times \sigma^{\prime} \right) \rangle \neq 0$ at
$x^\star$. Thus, the Jacobian of $\mathscr{T}$ is invertible in a \nbhd
of $x^{\star}$. By\cite[Proposition 1.2.3]{Isi95}, $\mathscr{T}$ defines a local diffeomorphism.
\end{proof}}
The equivalent transformed system after applying $\mathscr{T}$ to the extended system~\eqref{eq:extended_system} becomes
\vspace{-0.05in}
{\footnotesize
\begin{equation}
\label{eq:dynamic_quadrotor_new_coordinates}
\begin{aligned}
      \dot{\xi}_{ij} &= \xi_{ij+1}\\
      \dot{\eta}_{j} &= \eta_{j+1}\\
      \dot{\xi}_{14} &= L_{f}^{4}\alpha_1 + \left. L_{g_{1}}L_{f}^{3}\alpha_1 u_d + L_{g_{2}}L_{f}^{3}\alpha_1 u_r + L_{g_{3}}L_{f}^{3}\alpha_1 u_q  \right|_{x = T^{-1}(\xi,\eta,\mu)}\\
      \dot{\xi}_{24} &= L_{f}^{4}\alpha_2 + \left. L_{g_{1}}L_{f}^{3}\alpha_2 u_d+ L_{g_{2}}L_{f}^{3}\alpha_2 u_r+ L_{g_{3}}L_{f}^{3}\alpha_2 u_q \right|_{x = T^{-1}(\xi,\eta,\mu)}\\
      \dot{\eta}_{4} &= L_{f}^{4}\pi + \left.  L_{g_{1}}L_{f}^{3}\pi u_d +L_{g_{2}}L_{f}^{3}\pi u_r +L_{g_{3}}L_{f}^{3}\pi u_q \right|_{x = T^{-1}(\xi,\eta,\mu)}\\
      \dot{\mu}_{k} &= \left. b_k(\xi,\eta,\mu) \right|_{x = T^{-1}(\xi,\eta,\mu)}
\end{aligned}
\end{equation}}
for $i\in\{1,2\}$, $j\in\{1,2,3\}$, $k\in\{1,2\}$ and where $b_k$ are smooth maps. From~\eqref{eq:dynamic_quadrotor_new_coordinates}, one can define a regular feedback of the form
\begin{eqnarray}
\label{eq:regular_feedback_transformation}
 \left[\!\!\!
  \begin{array}{c}
      u_{d}\\
      u_{r}\\
      u_{q}\\
  \end{array} \!\!\!\right]:=D^{-1}(x)\left( \left[\!\!\!
  \begin{array}{c}
      -L_{f}^{4}\alpha_1\\
      -L_{f}^{4}\alpha_2\\
      -L_{f}^{4}\pi\\
  \end{array} \!\!\!\right]+
  \left[\!\!\!
  \begin{array}{c}
      v^{\xi_{1}}\\
      v^{\xi_{2}}\\
      v^{\eta}\\
  \end{array} \!\!\!\right]
  \right),
\end{eqnarray}
where $( v^{\xi_{1}},v^{\xi_{2}},v^{\eta} )$ are auxiliary control
inputs.
Using Lemma~\ref{lem:welldefined_deg} and the regular feedback law~\eqref{eq:regular_feedback_transformation}, the extended quadrotor model~\eqref{eq:extended_system} transformed to $3$ decoupled chain of integrators and two-dimensional internal dynamics is given as follows:
\begin{equation}
\left. \begin{array}{cccc}\label{eq:chain_form}
		\dot{\xi}_{11}={\xi_{12}} & \dot{\xi}_{21}={\xi_{22}} &  \dot{\eta}_{1}={\eta_{2}} & \dot\mu=b(\xi,\eta,\mu).
        \\
		\dot{\xi}_{12}={\xi_{13}} & \dot{\xi}_{22}={\xi_{23}} & \dot{\eta}_{2}={\eta_{3}} &
        \\
		\dot{\xi}_{13}={\xi_{14}} & \dot{\xi}_{23}={\xi_{24}}  & \dot{\eta}_{3}={\eta_{4}}  &
        \\
		\dot{\xi}_{14}= v^{\xi_{1}} & \dot{\xi}_{24}= v^{\xi_{2}} & \dot{\eta}_{4}= v^{\eta} &\\
		   \end{array} \right.
\end{equation}
Since~\eqref{eq:chain_form} is a partial linear system, designing a controller for the linear subsystems is straightforward. For the $\xi$-subsystem, we design the following linear controller:
\begin{equation}
\label{eq:linear-xi-controller}
\begin{aligned}
    v^{\xi_{1}} & = -\sum_{i=1}^{4}k^{\xi}_{1i}\xi_{1i},\quad
    v^{\xi_{2}}  = -\sum_{i=1}^{4}k^{\xi}_{2i}\xi_{2i},
\end{aligned}
\end{equation}
for some strictly positive gains $k^{\xi}_{1i}$ and $k^{\xi}_{2i}$, $i\in\{1,\cdots,4\}$, which renders the resulting closed-loop linear system Hurwitz. We underscore that by stabilizing the origin of the $\xi$-subsystem, the system state converges to the set $\Gamma^\star$.
}
Similarly, a linear control law is designed for the
$\eta-$subsystem
{
\begin{equation}
\label{eq:linear-eta-controller}
\begin{aligned}
    v^{\eta} & = - k^{\eta}_{2}(\eta_{2} - \dot\eta^{ref}_{1}) -k^{\eta}_{3}(\eta_{3} - \ddot\eta^{ref}_{1}) - k^{\eta}_{4}(\eta_{4} - \dddot\eta^{ref}_{1})
\end{aligned}
\end{equation}
for some appropriate gains $k^{\eta}_{i}$ for $i\in\{2,3,4\}$ to ensure stability.

Next, we show that the internal state $\mu\in\Real^2$ remains bounded under bounded adversarial attacks.

\subsection{Internal Dynamics}
\label{subsec:int_dyn}
The internal dynamics of the uncontrolled $\mu$-subsystem is as follows:
\begin{equation}
\label{eq:internal_dynamics}
\begin{aligned}
       \dot \mu_1(\xi, \eta, \mu) &= \dot \psi = \left. \frac{1}{{c_{\theta}}}\left(qs_{\phi} + rc_{\phi}\right) \right|_{x = T^{-1}(\xi, \eta, \mu)},\\
      \dot \mu_2(\xi, \eta, \mu) &= -\frac{l x_{13}}{2 I_x I_z} + \frac{l k_r r + pql(I_y - I_x)}{2 I_x I_z d} + \\
      & \left. \frac{k_r p + 2\mc{A}(t,x_q)l +qr(I_z - I_y)}{I_x I_z}\right|_{x = T^{-1}(\xi, \eta, \mu)}.
      \end{aligned}
    \end{equation}
We now prove the boundedness of internal dynamics:

}

{
\begin{lemma}
\label{lemma:Int_Dyn_Bounded}
If the quadrotor system has bounded control inputs $(u_f, u_r, u_q)\in\Real^3$ and Assumption \ref{as:uf_nonzero_&_no gimbal lock situation} holds, then the derivatives of the internal states, i.e., $\dot \mu_1$ and $\dot \mu_2$ given in~\eqref{eq:internal_dynamics} are bounded. Also, $\mu_2$ is bounded.
\end{lemma}
\begin{proof}
It can be shown that bounded control inputs imply bounded body rates $p,q,r$. By Assumption \ref{as:uf_nonzero_&_no gimbal lock situation}, the system is bounded away from singularities $\left(\phi, \theta = \pm \pi/2\right)$. From~\eqref{eq:internal_dynamics}, we get
\vspace{-0.1in}
\begin{equation}
\label{eq:int_dyn_bounded}
\begin{aligned}
| \dot \mu_1 | & \leq \left| \frac{1}{{c_{\theta}}} \right| \left(| q | + |r |\right),\\
     | \dot \mu_2 | & \leq \frac{l |x_{13}| }{2 I_x I_z} + \frac{l k_r
| r | + |pql(I_y - I_x)|}{2 I_x I_z d}  + \\
      &   \frac{k_r | p | + 2 |\mc{A}(t,x_q)| l +|qr(I_z - I_y)|}{I_x I_z},\\
| \mu_2 | & \leq \frac{|p|}{I_z} + \frac{l |r| }{2 I_x d},
\end{aligned}
\end{equation}
which are bounded because the body rates $p,q,r$, the thrust forces $f_1, f_3, f_4$, and the adversarial signal $\mc{A}$ are all bounded.
\end{proof}
}

We are now ready to prove the main result.
{
\begin{theorem}[Main Result]
\label{theo:main_result}
    Given the extended system~\eqref{eq:extended_system}, the virtual output function~\eqref{eq:virtual_output_fn}, an adversarial attack signal $\mc{A}$ defined in Definition \ref{def:adversarial_signal}, the regular feedback given in~\eqref{eq:regular_feedback_transformation} along with the controllers in~\eqref{eq:linear-xi-controller}, the path following manifold $\Gamma^\star$ given in~\eqref{eq:path_following_manifold} is asymptotically stable and forward-invariant.
\end{theorem}
\begin{proof}
The feedback transformation~\eqref{eq:regular_feedback_transformation} is well-defined in a \nbhd of the path-following manifold $\Gamma^{\star}$. By Lemma \ref{lemma:Diffeo}, the function $\mathscr{T}$ is a diffeomorphism of a \nbhd of $\Gamma^{\star}$ onto its image. In $(\xi, \eta, \mu)$ coordinates, the set $\Gamma^{\star}$ becomes $\mathscr{T}(\Gamma^{\star}) = \{(\xi, \eta, \mu) : \xi = 0 \}$. The $\xi$-subsystem in~\eqref{eq:chain_form}, after the feedback transformation, is linear time-invariant, and the feedback~\eqref{eq:linear-xi-controller} makes the set $\Gamma^\star$ exponentially stable. Thus, for any $x(0)$ close to $\Gamma^{\star}$, or equivalently, for any $\xi(0)$ close to 0, $\xi(t) \rightarrow 0$ exponentially. Therefore, $\mathscr{T}(\Gamma^{\star})$ is exponentially stable, and since $\mathscr{T}$ is a local diffeomorphism, this implies that $\Gamma^{\star}$ is exponentially stable as well and hence, is asymptotically stable and also forward-invariant.
\end{proof}
It is immediate to see that stability of $\Gamma^\star$ implies that the output of the system converges to the set $\gamma$, i.e., $\norm{y(t)}_{\gamma} \to 0$, as $t\to\infty$, and by construction, $\gamma$ becomes a forward invariant set. Hence, Problem~\ref{problem:globalinvariance} is solved in the presence of bounded adversarial attacks. Moreover, the controller~\eqref{eq:linear-eta-controller} makes the quadrotor follow a desired speed profile along the path, even in the presence of adversarial attacks.
}

}

\section{Numerical Simulation}
\label{sec:Num_Simul}
{
We assume, without loss of generality, that rotor two is under a bounded attack and the quadrotor has a weight-to-thrust ratio of $1:4$, i.e., each propeller is capable of producing at least $8$N of force~\cite{BauSca2021}. It is assumed that the rotors cannot produce negative thrust, hence the actuator limits are between zero and eight, i.e., $f_i(t) \in [0,8]$ for $i \in  \set{1,3,4}$. The mass of the quadrotor is $1.1$ kg and the inertias of the quadrotor are $I_x = 0.0020066 \text{ kg.m}^2, I_y = 0.0020038 \text{ kg.m}^2, I_z = 0.0013323 \text{ kg.m}^2$. The length of each arm of the quadrotor is $l = 0.2656$ m.

The desired path is a non-closed curve given by $\gamma = \set{y\in\Real^3: y_2 - \cos y_1 = 0, y_3 + 4\cos y_1 - 8= 0}$, and the adversarial signal is dependent on both time and system state $x_q = (x_I, y_I, z_I)$, namely
$
\mc{A}(t,x_q) = 0.1 t + \sin z_I.
$
We highlight that the system only knows the instantaneous value of $\mc{A}(t,x_q)$ using a speed sensor attached to each motor of the quadrotor and not the whole attack model. Using a known motor constant $k_m$ and the measured motor speed $\omega_m$, one can measure the force exerted at the $i$-th motor by $f_i = k_m\omega_m$. Moreover, we assume $10\%$ of parametric uncertainty in the quadrotor's inertia and $2\%$ of uncertainty in mass. A small uncertainty in mass is considered because the mass of the quadrotor can be accurately measured within grams. On the other hand, the system's inertia is difficult to measure accurately. Figure~\ref{fig:sine-following} shows that the quadrotor is following the curve starting from an initial position of $x_q = (1.5,0.3.0.5)$m despite parametric uncertainties.
\begin{figure}
    \centering
    \includegraphics[width = 6cm, height = 4cm]{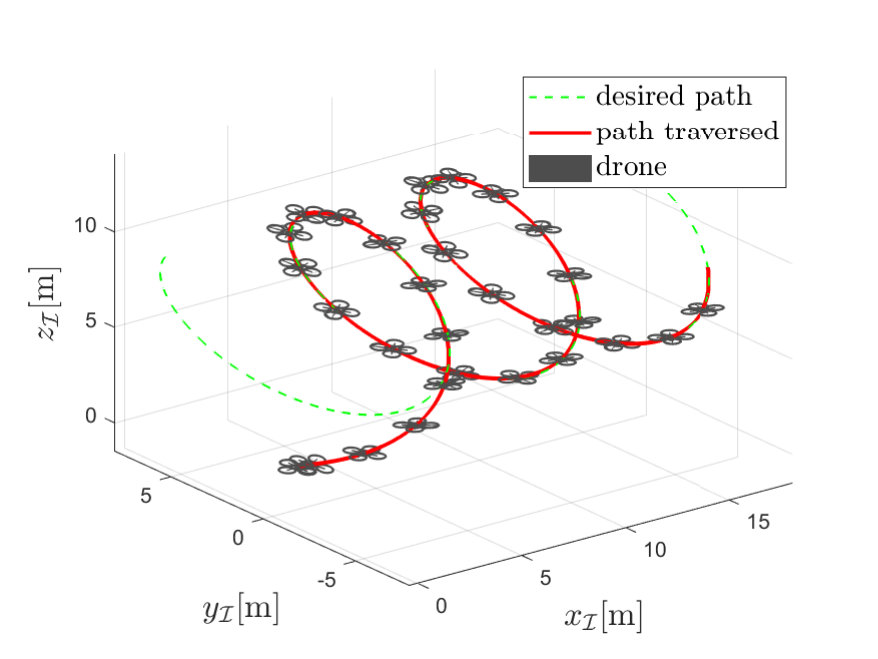}
    \caption{Quadrotor following the path}
    \label{fig:sine-following}
    \vspace{-0.2in}
\end{figure}

The bounded adversarial signal $\mc{A}(t,x_q)$ is plotted against time in the top plot of Figure~\ref{fig:sine-inputs}. It can also be seen in the remaining plots of Figure~\ref{fig:sine-inputs} that all the inputs remain bounded and well within the actuators' limit of $8$N. To watch the animation of this maneuver, click here.\footnote{\href{https://youtu.be/orp1g-pIFrM
}{https://youtu.be/orp1g-pIFrM}}

\begin{figure}[ht]
    \centering
    \setkeys{Gin}{width=0.475\linewidth}
\subfloat[Adversarial agent, input forces, torques, and thrust forces      \label{fig:sine-inputs}]{\includegraphics{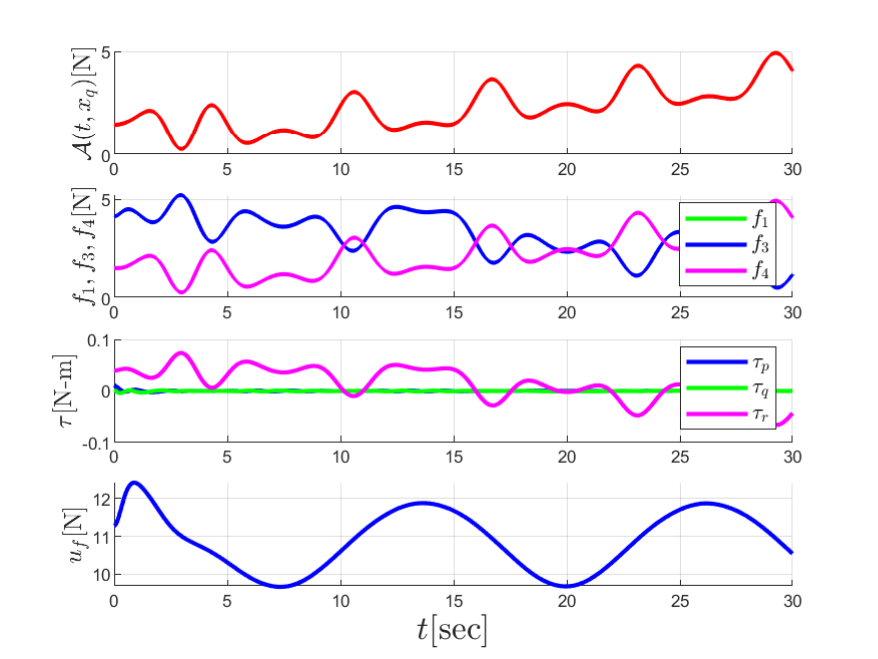} }\hfil
\subfloat[Quadrotor's velocities \label{fig:sine-velocities}]{\includegraphics{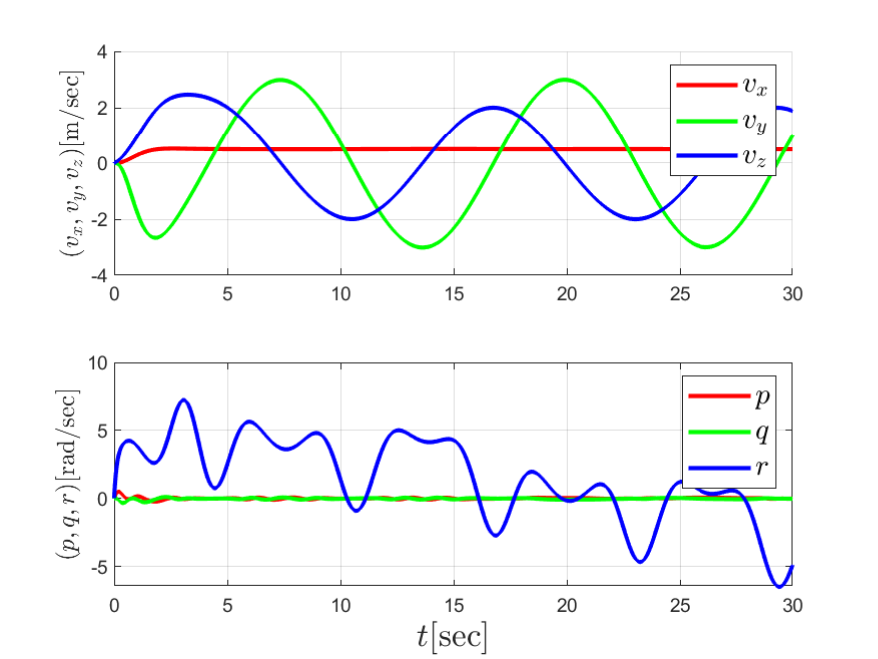} }

\caption{System's input and velocities}
\label{fig:first}
\end{figure}

Finally, Figure~\ref{fig:sine-velocities} shows the linear and angular velocities of the quadrotor, and it can be seen that since the adversarial agent is bounded, the body rates are also bounded.

In summary, the effectiveness of the proposed controller is shown through successful simulation. In the case of one rotor completely controlled by an adversarial agent, we have shown that even in the case of parametric uncertainties, the system can follow the given path accurately.

}

\bibliographystyle{IEEEtran}
\bibliography{myreferences}

\end{document}